\documentclass[letterpaper, 11pt]{article}
\usepackage[margin=1in,paperheight=11in,paperwidth=8.5in]{geometry}
\usepackage{amsthm}
\usepackage{amsmath,amssymb}
\usepackage[utf8]{inputenc} 
\usepackage[T1]{fontenc}    
\usepackage{hyperref}       
\usepackage{url}            
\usepackage{booktabs}       
\usepackage{amsfonts}       
\usepackage{nicefrac}       
\usepackage{microtype}      
\usepackage{palatino}
\usepackage{mathpazo}
\usepackage{graphicx}
\usepackage{color}
\usepackage[linesnumbered,ruled,vlined]{algorithm2e}
\SetKwInput{Input}{Input}
\SetKwInput{Output}{Output}
\usepackage{lipsum}

\newtheorem{thm}{Theorem}

\newtheorem{lem}{Lemma}
\newtheorem{dfn}{Definition}

\newtheorem{Cor}{Corollary}

\newtheorem{con}{Conjecture}
\def\<{\langle}
\def\>{\rangle}
\def\be{\begin{equation*}}
\def\ee{\end{equation*}}
\def\bea{\begin{eqnarray*}}
\def\eea{\end{eqnarray*}}

\newcommand{\comment}[1]{}

\newcommand{\ket}[1]{\ensuremath{\left|#1\right\rangle}}

\begin{document}

\title{Quantum Max-Flow Min-Cut theorem}

\author{Nengkun Yu\\
Centre for Quantum Software and Information, \\Faculty of Engineering and Information Technology,\\ University of Technology Sydney, NSW 2007, Australia}


\maketitle


\begin{abstract}
The max-flow min-cut theorem is a cornerstone result in combinatorial optimization. Calegari et al. (J Am Math Soc 23(1):107–188, 2010) initialized the study of quantum max-flow min-cut conjecture, which connects the rank of a tensor network and the min-cut. Cui et al. (J Math Phys 57:062206, 2016) showed that this conjecture is false generally. In this paper, we establish a quantum max-flow min-cut theorem for a new definition of quantum maximum flow. In particular, we show that for any quantum tensor network, there are infinitely many $n$, such that quantum max-flow equals quantum min-cut, after attaching dimension $n$ maximally entangled state to each edge as ancilla. Our result implies that the ratio of the quantum max-flow to the quantum min-cut converges to $1$ as the dimension $n$ tends to infinity. As a direct application, we prove the validity of the asymptotical version of the open problem about the quantum max-flow and the min-cut, proposed in Cui et al. (J Math Phys 57:062206, 2016).

\end{abstract}


\section{Introduction}
The max-flow min-cut theorem states that in a flow network, the maximum amount of flow passing from the source to the sink is equal to the total weight of the edges in a minimum cut \cite{ford_fulkerson_1956}. This beautiful theorem has nontrivial applications in network connectivity, network reliability, security of statistical data, data mining, distributed computing, and many many more \cite{10.5555/1942094}.

A tensor network associated with a graph is constructing tensors in large spaces from smaller
building-block tensors. Calegari et al. introduced the quantum max-flow min-cut conjecture in
\cite{10.2307/40587269}. This conjecture studies the rank of a tensor network by analysing the maximal classical flow (or minimal cut) on the graph corresponding to the tensor network.
Cui et al. considered two versions of the quantum max-flow min-cut conjecture and provided explicit counter-examples for both versions in \cite{Cui_2016}. They left an open problem about the asymptotical version of the quantum max-flow min-cut. They considered an interesting scenario with parameter $n$ as the degrees of freedom on the network's edges by attaching dimension $n$ maximally entangled state to each edge as ancilla for any quantum tensor network. They conjectured that the quantum max-flow becomes quantum min-cut as $n$ tends to infinity. Their intuition of the above equation is that a quantum phenomenon (quantum max-flow is strictly less than quantum min-cut) disappears in a more extensive system. Hastings proved that the quantum max-flow min-cut conjecture--version II is asymptotically true \cite{Hastings_2016}. More precisely, Hastings showed that the ratio of the quantum max-flow to the quantum min-cut converges to 1, in the limit of large dimension $n$. Hastings also suggested a weaker version of the quantum max-flow min-cut theorem: There exists $n$ such that version II quantum max-flow equals min-cut theorem. By verifying a numerical prediction from \cite{Hastings_2016}, Gesmundo et al. proved that there are infinitely many $n$ such that the version II quantum max-flow is strictly less than the quantum min-cut, even after attaching $n$-dimensional ancilla in \cite{Gesmundo_2018}. Regarding quantum tensor network as an entanglement network, we studied the quantum capacities for such quantum tensor network in \cite{7541586}.

This paper defines a new version of quantum maximum flow, motivated by the study of quantum capacity for the quantum repeater network with free classical communication. Due to the natural structure of quantum mechanics, the problem of finding a maximum flow for a quantum tensor network becomes integer programming with monomial constraints. We establish quantum max-flow min-cut theorem for this new definition of quantum maximum flow: For any quantum tensor network, there always exist infinitely many $n$, such that quantum max-flow equals quantum min-cut, after attaching dimension $n$ maximally entangled state to each edge as ancilla. This theorem implies that the ratio of the newly defined quantum max-flow and quantum min-cut converges to $1$ as the dimension $n$ tends to infinity. By connecting this quantum maximum flow and the version I quantum maximum flow in \cite{Cui_2016}, we show that the ratio of the version I quantum maximum flow to the quantum min-cut converges to $1$ as $n$ tends to infinity.

In Section \ref{preliminary}, we present definitions of quantum maximum flows and minimum cut. In Section \ref{qmaxmin}, we first prove the existence of rational flow equals minimum cut, then use that to show the quantum maximum flow minimum cut theorem. In section \ref{om}, we provide an operational meaning by interpreting our quantum maximum flow minimum cut theorem as the existence of entanglement distribution using quantum teleportation. That interpretation implies the validity of the asymptotical version of the conjecture in \cite{Cui_2016}.

\section{Preliminaries and Notations} \label{preliminary}
In this section, we provide the definitions of the quantum maximum flows and minimum cut after recalling the max-flow min-cut theorem of classical flow networks. Then we 

\subsection{Classical max-flow min-cut theorem}

The max-flow min-cut theorem relates two quantities: the maximum flow through a network and the minimum capacity of a cut of the network. That is, the flow achieves the minimum capacity. To state the theorem, we recall the definitions of these quantities.

Let $\mathcal{N}=(V,E)$ be a directed graph, where $V$ denotes the set of vertices, and $E$ is the set of edges. Let $S \subseteq V$ and $T\subseteq V$ be the source and the sink of $\mathcal{N}$, respectively.
We only need to consider networks with one source vertex $s$ and one sink vertex $t$. The general case for transmitting information
from $S$ to $T$ reduces to this simple case by viewing all $s_i$/$t_j$
as one source/sink.

The capacity of an edge is a mapping $c:E \mapsto\mathbb{R}^+$ denoted by $c_{u,v}$ where $u,v \in V$. It represents the maximum amount of flow that can pass through an edge.

The flow of a flow network is defined as follows. We will call it a classical flow.
\begin{dfn} \label{classical-flow}
A flow is a mapping $f : E \mapsto \mathbb{R}^+$ denoted by$f_{u,v}$, subject to the following two constraints:
\begin{itemize}
\item Capacity Constraint: For every edge $(u , v)\in E$, $f_{u,v} \leq c_{u,v}$.

\item Conservation of Flows: For each vertex $v$ apart from $s$ and $t$ (i.e. the source and sink, respectively), the following equality holds:
\begin{align*}
\sum_{ u : ( u , v ) \in E } f_{u v} = \sum_{ w : ( v , w ) \in E } f_{v, w} .
\end{align*}
\end{itemize}
The value of a flow is defined by
\begin{align*}
| f | = \sum_{ v : ( s , v ) \in E } f_{s, v} = \sum_{ v : ( v , t )\in E } f_{v ,t}.
\end{align*}

The maximum flow problem asks for the largest flow on a given network, denoted as $MF(\mathcal{N},c)$.
\end{dfn}

The cut of a flow network is the smallest total weight of the edges which if removed would disconnect the source from the sink. We will call it a classical cut.
\begin{dfn}\label{classical-cut}
A cut is a partition of $V$ into $\bar{S} \cup \bar{T}$ , where $s \in \bar{S}$ and $t \in\bar{T}$. The
cut set of the cut is the set of edges $(v,w)\in E$ with $v \in \bar{S}$ and $w \in\bar{T}$. For a given flow network, and a given cut set $\bar{S} \cup \bar{T}$ separating $s$ from $t$, its capacity is defined as
\begin{align*}
\mathop{\sum}\limits_{\begin{subarray}{c} u,v: (u,v)\in E \\u\in \bar{S}, v\in \bar{T}\end{subarray}} c_{u,v}.
\end{align*}
The min-cut of the network, $MC$ is the
minimum capacity over all cuts.
\end{dfn}

The following result links the maximum flow through a network with the minimum cut of the network. It was published in 1956 by L. R. Ford Jr. and D. R. Fulkerson.

\begin{thm}\label{classicalmaxmin} \cite{ford_fulkerson_1956}
The maximum value of a flow is equal to the minimum capacity overall cuts. Moreover, if the capacity for all edges in the graph are integers, then the Ford-Fulkerson method finds a max flow in which every flow value is an integer.
\end{thm}
We call the integer part of this theorem the Flow Integrality Theorem.

\subsection{Quantum maximum flows and min cut}

A tensor network associated with a graph is constructing tensors in large spaces from smaller
building-block tensors. Given an undirected graph $G = (V,E)$ with a capacity function
$d: E \rightarrow \mathbb{N}_{\geq 2}$ and a resource
$S=\{s\} \subseteq V$ (resp. sink $T=\{t\} \subseteq V$), we can define a quantum tensor network
associated to $(G,d,S,T)$ as an entanglement network,
where each pair of connected nodes of $e \in E$ share a maximally
entangled state
$\ket{\psi_{e}}=\sum_{i=1}^{d_{e}}\ket{ii}/\sqrt{d_{e}}$.

Usually, the dimension $d_e$ of the maximally entangled state $\ket{\psi_{e}}$ is
chosen to be the same on each edge $e\in E$ as studied in \cite{Hastings_2016}. However, the more general case, where the $d_e\,'$s may be different, is also known to be interesting, which has demonstrated connections to the topological quantum field theory and the theory of quantum gravity \cite{Cui_2016,Ryu_2006,Headrick_2014,Hayden_2016}.

We can define the quantum minimum cut as follows.

\begin{dfn}
A cut is a partition of $V$ into $\bar{S} \cup \bar{T}$ , where $s \in \bar{S}$ and $t \in\bar{T}$. The
cut set of the cut is the set of edges $(v,w)\in E$ with $v \in \bar{S}$ and $w \in\bar{T}$. For a given graph,
and a given cut set separating $s$ from $t$, define $D_C$ to equal the product of the capacities
of the edges in the cut set. In other words,
\begin{align*}
D_C=\mathop{\Pi}\limits_{\begin{subarray}{c} u,v: (u,v)\in E \\u\in \bar{S}, v\in \bar{T}\end{subarray}} d_{u,v}.
\end{align*}
The quantum min-cut of the network, $QMC$ is the
minimum of $D_C$ over all cuts.
\end{dfn}

Different from classcial min-cut, the quantum min-cut is defined as the product of the capacities
of the edges in the cut set. This quantum feature can be understood by observing 
\begin{align*}
\ket{\psi_{e}}\otimes \ket{\psi_{e'}}=\sum_{i=1}^{d_{e}d_{e'}}\ket{ii}/\sqrt{d_{e}d_{e'}}.
\end{align*}

An interesting type of entanglement networks we focus on is the tensor network, which transports linear algebraic things like rank and entanglement \cite{bauer2011tensor,Cui_2016}. Given a network $\mathcal{N} = (G, d, S,T)$, let $\tilde{V} = V \setminus (S \cup T)$, and for each $v \in V$, let $E(v) \subset E$ be the set of edges containing $v$. We define $V(S) = \bigotimes\limits_{e \in E(s)} \mathbb{C}^{d_e}$, and define $V(T)= \bigotimes\limits_{e \in E(t)} \mathbb{C}^{d_e}$. One can assign a set of tensors $\mathcal{T} = \{\mathcal{T}_u: u \in \tilde{V}\}$ to $\mathcal{N}$, where $\mathcal{T}_u$ is an arbitrary tensor in $\bigotimes\limits_{e \in E(u)}\mathbb{C}^{d_e}$, i.e., each index of $\mathcal{T}_u$ corresponds to an edge containing $u$. Then contracting the tensors along all internal edges results in a linear map $\beta_{\mathcal{T}}: V(S) \rightarrow V(T)$. With the notations above,
\begin{dfn}
The quantum max-flow, $\widetilde{QMF}(\mathcal{N})$, is the the maximum rank of $\beta(\mathcal{N};\mathcal{T})$ over all tensor assignments $\mathcal{T}$.
\end{dfn}

One can interpret the quantum max-flow $\widetilde{QMF}(\mathcal{N})$ as the maximum rank of entanglement between the source and the sink that can be generated by applying stochastic local operation assisted by classical communication protocols. One can implement $\mathcal{T}_u$ in $\bigotimes\limits_{e \in E(u)}\mathbb{C}^{d_e}$ stochastically.

It is proved in \cite{Cui_2016} that 
\begin{thm}\label{exist}
For any quantum tensor network $\mathcal{N}=(G,d,S,T)$,
\begin{align*}
\widetilde{QMF}(\mathcal{N})\leq QMC(\mathcal{N}).
\end{align*}
\end{thm}
Unfortuantely, the equality is not generally saturated. For given quantum tensor network, they studied new networks by attaching an additional maximally entangled state $\ket{\psi}=\sum_{i=1}^n \ket{ii}/{\sqrt{n}}$ on each edge.
\begin{dfn}
Given the network $\mathcal{N} = (G, d, S, T)$, and $n\in\mathbb{N}$, we define $n\mathcal{N} = (G, nd, S, T)$. In other words, the capacity of each edge multuples $n$. 
\end{dfn}
Motivated by properties of frustration free local hamiltonian, Cui et al. proposed the open question
\begin{con}
For any quantum tensor network $\mathcal{N}=(G,d,S,T)$,
\begin{align*}
\forall \mathcal{N}, \ \ \  \lim_{n \to \infty} \widetilde{QMF}(n\mathcal{N})=QMC(n\mathcal{N}).
\end{align*}
\end{con}

To track this problem, we consider the quantum information transmission from the source vertex $s$ to the
sink vertex $t$ through the network via local quantum
operations and classical communications. More precisely, we use quantum teleportation for transmitting quantum information in the network. One can implement quantum teleportation by allowing quantum operations at each vertex, and classical communications are free. The goal is then to establish maximum
bipartite entanglement between $S$ and $T$ via teleportation protocols.

Since we allow unlimited classical communications, the direction of quantum communication on each edge $e$ can be arbitrary given a maximally entangled state associated with an
edge $e$. Therefore, we allow both directions of quantum communication for any edge $e$. The only restriction is that the product of both capacities is at most $d_e$.

This motivates us to provide the following definition of the quantum flow.
\begin{dfn}\label{quantum-flow}
A flow is a map $f:E\mapsto \mathbb{N}$ that satisfies the following:
\begin{itemize}
\item Capacity constraint. The flow of an edge cannot exceed its capacity, in other words: $1\leq f_{v,u}\times f_{u,v}\leq d_{u,v}$ for all $(u,v)\in E$.
\item Conservation of flows. The product of the flows entering a node must equal the product of the flows exiting that node, except for the source and the sink. In other words, $\forall v\in V\setminus\{s,t\}$:
\begin{align}
\mathop{\Pi}\limits_{u:(u,v)\in E} {f_{u,v}} =\mathop{\Pi}\limits_{u:(v,u)\in E} f_{v,u}.
\end{align}
\end{itemize}
\end{dfn}
We allow the incoming flow of the source and outgoing flow of the sink. In other words, there may exists $v$ such that $f_{v,s}>1$ or $f_{t,v}>1$.

The value of a flow $f$ is defined as follows.
\begin{dfn}
The value of a flow is the amount of flow passing from the source to the sink. Formally for a flow $f:E\mapsto \mathbb{N}$ it is given by
\begin{align}
|f|=\frac{\mathop{\Pi}\limits_{v:(s,v)\in E} f_{s,v}}{\mathop{\Pi}\limits_{u:(u,s)\in E} f_{u,s}}.
\end{align}
The maximum flow problem is to route as much flow as possible from the source to the sink. In other words, the maximum flow of the network, $QMF$, is the flow $f_{\max}$ with maximum value.
\end{dfn}

A particular flow class with no incoming flow of the source and no outgoing flow of the sink is of great interest.
\begin{dfn}\label{restrict-quantum-flow}
A flow $f:E\mapsto \mathbb{N}$ is a strict flow if $f_{v,s}=f_{t,u}=1$ for all $(s,v),(u,t)\in E$. The maximum  strict flow of the network, $QMF_s$, is the flow $f_{\max}$ with maximum value.
\end{dfn}
Here $G = (V,E)$ is undirected graph, therefore, $(u,v)\in E$ implies $(v,u)\in E$.

Compute the quantum maximum flow and strict quantum maximum flow are both 

To prove our main results, we need the following definitions of $\mathbb{R}$ flow and $\mathbb{Q}$ flow. For $\mathbb{F}=\mathbb{R}$ or $\mathbb{Q}$, we define the quantum flow over $\mathbb{F}$.
\begin{dfn}
For quantum tensor network $\mathcal{N}=(G,d,S,T)$, a flow over $\mathbb{F}$ is a map $f:E\mapsto \mathbb{F}$ that satisfies the following:
\begin{itemize}
\item Single direction. $f_{v,u}=1$ or $f_{u,v}=1$ for all $(u,v)\in E$.
\item Capacity constraint. The flow of an edge cannot exceed its capacity, in other words: $\frac{1}{d_{u,v}}\leq f_{v,u}, f_{u,v}\leq d_{u,v}$ for all $(u,v)\in E$.
\item Conservation of flows. The product of the flows entering a node must equal the product of the flows exiting that node, except for the source and the sink. In other words, $\forall v\in V\setminus\{s,t\}$:
\begin{align}
\mathop{\Pi}\limits_{u:(u,v)\in E} {f_{u,v}} =\mathop{\Pi}\limits_{u:(v,u)\in E} f_{v,u}.
\end{align}
\end{itemize}
The value of a flow $f:E\mapsto \mathbb{F}$ it is given by
\begin{align}
|f|=\frac{\mathop{\Pi}\limits_{v:(s,v)\in E} f_{s,v}}{\mathop{\Pi}\limits_{u:(u,s)\in E} f_{u,s}}.
\end{align}
The maximum (or supremum) flow over $\mathbb{F}$ of the network, $QMF_{\mathbb{F}}$, is the flow with maximum (or supremum) value over $\mathbb{F}$. 
\end{dfn}
For $\mathbb{F}=\mathbb{R}$, the constraints are all closed, so the maximum is achievable. For $\mathbb{F}=\mathbb{Q}$, the definition uses supremum.

\section{Quantum maximum flow minimum cut theorem} \label{qmaxmin}

Before presenting the main result of this paper, we first observe the following lemma.
\begin{lem}\label{flow}
For any flow $f$ over $\mathbb{F}$ of the quantum tensor network $\mathcal{N} = (G, d, S, T)$, and any cut $\bar{S}$ and $\bar{T}$, we have
\begin{align*}
|f|=\frac{\mathop{\Pi}\limits_{\begin{subarray}{c} (v,w)\in E\\v\in \bar{S}, w\in \bar{T}\end{subarray}} f_{v,w}}{\mathop{\Pi}\limits_{\begin{subarray}{c} (r,u)\in E\\ u\in \bar{S}, r\in \bar{T}\end{subarray}} f_{r,u}}.
\end{align*}
\end{lem}
\begin{proof}
Let $W=\bar{S}\setminus \{s\}$.
\begin{align*}
|f|=&\frac{\mathop{\Pi}\limits_{v:(s,v)\in E} f_{s,v}}{\mathop{\Pi}\limits_{u:(u,s)\in E} f_{u,s}}\\
=&\frac{\mathop{\Pi}\limits_{v:(s,v)\in E} f_{s,v}}{\mathop{\Pi}\limits_{u:(u,s)\in E} f_{u,s}}\times {\mathop{\Pi}\limits_{x\in W}} \frac{\mathop{\Pi}\limits_{y:(x,y)\in E} f_{x,y}}{\mathop{\Pi}\limits_{z:(z,x)\in E} f_{z,x}}\\
=&\frac{\mathop{\Pi}\limits_{\begin{subarray}{c} (v,w)\in E\\v\in \bar{S}, w\in \bar{T}\end{subarray}} f_{v,w}}{\mathop{\Pi}\limits_{\begin{subarray}{c} (r,u)\in E\\ u\in \bar{S}, r\in \bar{T}\end{subarray}} f_{r,u}}.
\end{align*}
In the second line, we use the conservation of flows. In the third line, we use the following fact: The flow of each edge except those between $\bar{S}$ and $\bar{T}$ appears twice: One in the numerator and the other in the denominator.
\end{proof}
Using this lemma, we can show
\begin{lem}
For any $\mathcal{N} = (G, d, S, T)$,
\begin{align*}
QMF_{\mathbb{Q}}(\mathcal{N})\leq QMF_{\mathbb{R}}(\mathcal{N})\leq QMC(\mathcal{N}).
\end{align*}
\end{lem}
\begin{proof}
Observe that a flow over $\mathbb{Q}$ is also a flow over $\mathbb{R}$, we know that
\begin{align*}
QMF_{\mathbb{Q}}(\mathcal{N})\leq QMF_{\mathbb{R}}(\mathcal{N}).
\end{align*}
To bound $QMF_{\mathbb{R}}(\mathcal{N})$, we choose a minimum cut $\bar{S}\cup\bar{T}$ with $C=(\bar{S}\times \bar{T})\subseteq E$. Using Lemma \ref{flow}, we know that any flow $f$ over $\mathbb{R}$ satsifies
\begin{align*}
|f|=\frac{\mathop{\Pi}\limits_{\begin{subarray}{c} (v,w)\in E\\v\in \bar{S}, w\in \bar{T}\end{subarray}} f_{v,w}}{\mathop{\Pi}\limits_{\begin{subarray}{c} (r,u)\in E\\ u\in \bar{S}, r\in \bar{T}\end{subarray}} f_{r,u}}
=\frac{\mathop{\Pi}\limits_{ (v,w)\in C} f_{v,w}}{\mathop{\Pi}\limits_{(u,r)\in C} f_{r,u}}
={\mathop{\Pi}\limits_{ (v,w)\in C}} \frac{f_{v,w}}{f_{w,v}}
\leq {\mathop{\Pi}\limits_{ (v,w)\in C}} d_{v,w}
=QMC(\mathcal{N}),
\end{align*}
where in the last step, we use the Single direction condition and the capacity constraint of a flow: if $f_{w,v}=1$, then $\frac{f_{v,w}}{f_{w,v}}=f_{v,w}\leq d_{v,w}$; otherwise $f_{v,w}=1$, then
 $\frac{f_{v,w}}{f_{w,v}}=\frac{1}{f_{w,v}}\leq d_{v,w}$.
\end{proof}
This proof of the above lemma shows the following property of flows which saturates the minimum cut.
\begin{lem}\label{saturate}
If a flow $f$ over $\mathbb{R}$ in quantum tensor network $\mathcal{N}$ such that $|f|=QMC(\mathcal{N})$, then for any minimum cut $\bar{S}\cup\bar{T}$ with $C=(\bar{S}\times \bar{T})\subseteq E$, we have
$f_{u,v}=d_{u,v}$ or $f_{v,u}=\frac{1}{d_{u,v}}$ for $u\in \bar{S},v\in\bar{T}$ and $(u,v)\in E$.
\end{lem}

Interestingly, there is always a rational flow to achieve the quantum minimum cut.
\begin{lem}\label{rational}
For any quantum tensor network $\mathcal{N} = (G, d, S, T)$, 
\begin{align*}
QMF_{\mathbb{Q}}(\mathcal{N})=QMC(\mathcal{N}).
\end{align*}
Moreover, there is a rational flow which achieves the maximum flow.
\end{lem}
\begin{proof}
We first prove that 
\begin{align*}
QMF_{\mathbb{R}}(\mathcal{N})=QMC(\mathcal{N}).
\end{align*}
To show this, we only need to provide a flow $f$ over $\mathbb{R}$ satisfies $|f|=QMC(\mathcal{N})$. We consider a classical flow network $\mathcal{M} = (G', c, S, T)$ from quantum tensor network $\mathcal{N} = (G, d, S, T)$ with undirected graph $G=(V,T)$: $G'=(V,E')$ is a directed graph where 
\begin{itemize}
\item $E'=E\cup \{(u,v)| (v,u)\in E,\ u\neq s, v\neq t\}$, 
\item the capacity $c_{u,v}=c_{v,u}=\log d_{u,v}, \forall (u,v)\in E \ \mathrm{or}\ (v,u)\in E$. 
\end{itemize}
In other words, the source $s$ has only outgoing edges, the sink has only incoming edges, and other edges in $\mathcal{N}$ are duplicated in both directions.

For any cut of $\mathcal{M}$, $V=\bar{S} \cup \bar{T}$ with $s \in \bar{S}$ and $t \in\bar{T}$ as defined in Definition \ref{classical-cut}. The value of this cut is
\begin{align*}
\mathop{\sum}\limits_{\begin{subarray}{c} u,v: (u,v)\in E' \\u\in \bar{S}, v\in \bar{T}\end{subarray}} \log d_{u,v}=\log(\mathop{\Pi}\limits_{\begin{subarray}{c} u,v: (u,v)\in E' \\u\in \bar{S}, v\in \bar{T}\end{subarray}}  d_{u,v}).
\end{align*}
It is clear to observe that the following relation between the min-cut of $\mathcal{M}$ and $\mathcal{N}$.
\begin{align*}
\exp(MC(\mathcal{M}))=QMC(\mathcal{N}).
\end{align*}

We consider classical flow over $\mathcal{M}$ as Definition \ref{classical-flow}. Theorem \ref{classicalmaxmin}, the max-flow min-cut theorem, implies that there exists a flow $h: E'\mapsto \mathbb{R}^+$ such that
\begin{align*}
0\leq& h_{u,v} \leq \log d_{u,v},\\
\sum_{ u : ( u , v ) \in E } h'_{u, v} =& \sum_{ w : ( v , w ) \in E } h'_{v, w}, \forall\ v\notin\{s,t\}\\
\mathop{\sum}\limits_{v:(s,v)\in E} h_{s,v}=&\min \mathop{\sum}\limits_{\begin{subarray}{c} u,v: (u,v)\in E \\u\in \bar{S}, v\in \bar{T}\end{subarray}} \log d_{u,v}.
\end{align*}
It is possible that there exist $u,v\in E'$ such that both $h_{u,v},h_{v,u}>0$. We can simplify the flow $h$ to $h'$ as following:
For $u,v$ such that both $h_{u,v},h_{v,u}>0$, we choose to let $h'_{u,v}=h_{u,v}-h_{v,u}$ and $h'_{v,u}=0$. We have
\begin{align*}
-\log d_{u,v}\leq& h'_{u,v}\leq\log d_{u,v}\\
\sum_{ u : ( u , v ) \in E } h_{u v} =& \sum_{ w : ( v , w ) \in E } h_{v, w},\forall\ v\notin\{s,t\} \\
\mathop{\sum}\limits_{v:(s,v)\in E} h'_{s,v}=&\min \mathop{\sum}\limits_{\begin{subarray}{c} u,v: (u,v)\in E \\u\in \bar{S}, v\in \bar{T}\end{subarray}} \log d_{u,v}.
\end{align*}
Now, we have for all $(u,v)\in E'$, $h'_{u,v}=0$  or  $h'_{v,u}=0$.

Define $f:E\mapsto \mathbb{R}^+$ where $f_{u,v}=\exp(h'_{u,v})$. According to the properties of $h'$, we know that $|f|=QMC(\mathcal{N})$ and $\forall (u,v)\in E'$,
\begin{align*}
\frac{1}{d_{u,v}}\leq& f_{u,v} \leq d_{u,v},\\
\mathop{\Pi}\limits_{ u : ( u , v ) \in E }f_{u, v} =& \mathop{\Pi}\limits_{ w : ( v , w ) \in E } f_{v, w}, \forall\ v\notin\{s,t\}\\
\mathop{\Pi}\limits_{v:(s,v)\in E} f_{s,v}=&\min \mathop{\Pi}\limits_{\begin{subarray}{c} u,v: (u,v)\in E \\u\in \bar{S}, v\in \bar{T}\end{subarray}}  d_{u,v}=QMC(\mathcal{N}),\\
f_{u,v}=&1\ \  \mathrm{or}\ \  f_{v,u}=1.
 \end{align*}
 In other words, $f$ is a flow over $\mathbb{R}$ in quantum tensor network $\mathcal{N}$ and satisfies
 \begin{align*}
 QMC(\mathcal{N})=|f|\leq QMF_{\mathbb{R}}(\mathcal{N}) \leq QMC(\mathcal{N}).
 \end{align*}
If $f$ is also a rational flow, then the statement of this lemma is already true.
 
Otherwise, there exists at least one $u,v\in V$ such that $f_{u,v}\notin \mathbb{Q}$. We observe that there exists $w\neq v$ such that $f_{w,u}\notin \mathbb{Q}$ or $f_{u,w}\notin \mathbb{Q}$. To see this, we consider three cases:

Case 1: $u\notin \{s,t\}$: Since $f_{u,v}$ is a flow, according to the conservation condition on $u$
\begin{align}
\mathop{\Pi}\limits_{w:(w,u)\in E} {f_{w,u}} =\mathop{\Pi}\limits_{x:(u,x)\in E} f_{u,x}.
\end{align}
The product of the right handside $\mathop{\Pi}\limits_{x:(u,x)\in E} f_{u,x}$ contains $f_{u,v}\notin \mathbb{Q}$. If the product $\mathop{\Pi}\limits_{x:(u,x)\in E} f_{u,x}\in \mathbb{Q}$, then there is an $w\neq v$ such that  $f_{u,w}\notin \mathbb{Q}$; If the product $\mathop{\Pi}\limits_{x:(u,x)\in E} f_{u,x}\notin \mathbb{Q}$, then there is an $w\neq v$ such that  $f_{w,u}\notin \mathbb{Q}$.

Case 2: $u=s$. According to
\begin{align*}
|f|=\mathop{\Pi}\limits_{v:(s,v)\in E} f_{s,v}=QMC(\mathcal{N})\in \mathbb{Q},
\end{align*}
there exists $w\not v$ such that $f_{s,w}\notin \mathbb{Q}$.
 
Case 3: $u=t$. According to Lemma \ref{flow}, we have 
\begin{align*}
|f|=\mathop{\Pi}\limits_{v:(v,t)\in E} f_{v,t}=QMC(\mathcal{N})\in \mathbb{Q}.
\end{align*}
Thus, there exists $w\not v$ such that $f_{w,t}\notin \mathbb{Q}$.

By repreating this procedure, we can obtain a circle $u_0,\cdots,u_m$ since the graph $G$ is finite. More precisely, we have
\begin{align*}
f_{u_0,u_{m}}\notin \mathbb{Q}\ \mathrm{or}\ f_{u_{m},u_0}\notin \mathbb{Q}\ \ 
f_{u_i,u_{i+1}}\notin \mathbb{Q} \ \mathrm{or}\  f_{u_{i+1},u_i}\notin \mathbb{Q}  \ \ \forall 0\leq i\leq m-1
\end{align*}

For any $(y,z)\in E$, we can define $f'$ by 
\begin{align*}
f'_{y,z}=\frac{1}{f_{z,y}},\\
f'_{z,y}=\frac{1}{f_{y,z}}
\end{align*}
and leaving the rest of $f$ unchanged. 
By using this method to modify $f_{u_i,u_{i+1}}$ and $f_{u_{i+1},u_i}$ for each $i$, we can obtain a flow $f':E''\mapsto $ satisfies
\begin{itemize}
\item  $f_{v,u}=1$ or $f_{u,v}=1$ for all $(u,v)\in E'$.
\item  $\frac{1}{d_{u,v}}\leq f_{v,u}, f_{u,v}\leq d_{u,v}$ for all $(u,v)\in E'$.
\item  $\mathop{\Pi}\limits_{u:(u,v)\in E} {f_{u,v}} =\mathop{\Pi}\limits_{u:(v,u)\in E} f_{v,u}$ holds $\forall v\in V\setminus\{s,t\}$.
\item $|f'|=QMC(\mathcal{N})$.
\item  $u_0,u_1,\cdots, u_m$ such that $f'_{u_i,u_{i+1}},f'_{u_{m},u_0}\notin \mathbb{Q}$ for all $0\leq i\leq m-1$. 
\end{itemize}
where $E''=E\cup \{(u,v)| (v,u)\in E\}$. In otherwords, we allow incoming edges of $s$ and outgoing edges of $t$.
According to lemma \ref{saturate}, we know that any edge of $(u_m,u_0),(u_0,u_m),(u_i,u_{i+1}),(u_{i+1},u_i)$ are not in any minmimum cut.

Then, for all $0\leq i\leq m-1$. 
\begin{align*}
\frac{1}{d_{u_i,u_{i+1}}}<f'_{u_i,u_{i+1}}<d_{u_i,u_{i+1}}\\
\frac{1}{d_{u_m,u_0}}<f'_{u_m,u_0}<d_{u_m,u_0}.
\end{align*}

By choose some $\epsilon$ such that $(1+\epsilon)f'_{u_0,u_{1}}\in\mathbb{Q}$, we design $f'': E''\mapsto \mathbb{R}^+$ by
\begin{align*}
\frac{1}{d_{u_i,u_{i+1}}}<f''_{u_i,u_{i+1}}:=(1+\epsilon)f'_{u_i,u_{i+1}}<d_{u_i,u_{i+1}}\\
\frac{1}{d_{u_m,u_0}}<f''_{u_m,u_0}:=(1+\epsilon)f'_{u_m,u_0}<d_{u_m,u_0}.
\end{align*}
and leave the rest of $f'$ unchange.
One can verify that $f''$ is still a flow over $\mathbb{R}$ with $|f''|=QMC(\mathcal{N})$ according to Lemma \ref{flow}.

Moreover, each run of this procedure decreases the number of non-rational edge asignement by at least $1$. 

By repeating this procedure, we can obtain a flow $g: E''\mapsto \mathbb{Q}$ such that $|g|=QMC(\mathcal{N})$, and $\forall (u,v)\in E$,
\begin{align*}
\frac{1}{d_{u,v}}\leq f_{u,v} \leq d_{u,v},\\
f_{u,v}=1\ \  \mathrm{or}\ \  f_{v,u}=1.
 \end{align*}
 This proves
 \begin{align*}
 QMF_{\mathbb{Q}}(\mathcal{N})=QMC(\mathcal{N}).
 \end{align*}
\end{proof}

Using Lemma \ref{rational}, we are able to prove the following lemma.
\begin{lem}\label{main}
For any $\mathcal{N} = (G, d, S, T)$, there are infinitely many $n\in\mathbb{N}$ such that
\begin{align*}
QMF(n\mathcal{N})=QMC(n\mathcal{N}).
\end{align*}
\end{lem}
\begin{proof}
For any cut $\bar{S}\cup\bar{T}$, the value of the cut in $n\mathcal{N}$ is
\begin{align*}
n^{|(\bar{S}\times\bar{T})\cap E|}\times \mathop{\Pi}\limits_{(u,v)\in (\bar{S}\times\bar{T})\cap E } d_{u,v},
\end{align*}
where $|(\bar{S}\times\bar{T})\cap E|$ denotes the size of $(\bar{S}\times\bar{T})\cap E$.

Then, there exists $n_0$ such that for any $n\geq n_0$, a minimum cut of $n\mathcal{N}$ is also a minimum cut of $n_0\mathcal{N}$. Let us choose a minimum cut of $n_0{\mathcal{N}}$, denoted by $C:= (\bar{S}_0\times\bar{T}_0)\cap E$. Therefore, for $n\geq n_0$,
\begin{align*}
QMC(n\mathcal{N})= n^{|C|}\times\mathop{\Pi}\limits_{(u,v)\in C} d_{u,v},
\end{align*}
with $|C|$ being the size of $C$.

According to Lemma \ref{rational}, there exists a flow $g: E\mapsto \mathbb{Q}$ satisfies $|g|=QMC(n_0\mathcal{N})$ and $\forall (u,v)\in E$,
\begin{align*}
\frac{1}{d_{u,v}}\leq g_{u,v},g_{v,u} \leq d_{u,v},\\
g_{u,v}=1\ \ \mathrm{or}\ \ g_{v,u}=1.
\end{align*}
Let $m_0\in \mathbb{N}$ satisfies that for all $(u,v)\in E$
\begin{align*}m_0 g_{u,v},m_0 g_{v,u}, \frac{m_0}{g_{u,v}}, \frac{m_0}{g_{v,u}}\in\mathbb{N}.
\end{align*}

In the following, we will prove that for any sufficient large $k\in\mathbb{N}$, there exists a flow $f: E\mapsto \mathbb{N}$ of $kn_0m_0\mathcal{N}$ such that
\begin{align*}
|f|=(kn_0m_0)^{|C|}\mathop{\Pi}\limits_{(u,v)\in C} d_{u,v} =QMC(kn_0m_0\mathcal{N}).
\end{align*}

We first define a classical flow network $\mathcal{M}=(G,d',S,T)$ from $\mathcal{N}$ with $d'_{u,v}=1$ if $(u,v)\in E$. In such flow network as Definition \ref{classical-flow}, the min-cut of $\mathcal{M}$ is $|C|$.

The integrality theorem of Theorem \ref{classicalmaxmin} states that if all arc capacities are integer, the maximum flow problem has an integer maximum flow. 

Applying the integrality theorem on $\mathcal{M}$, we know that there is a flow $h:E\mapsto \{0,1,-1\}$ with value $|C|$. We can first simplify that into an equivalent flow $h':E\mapsto \{0,1\}$. In other words, we can choose $|C|$ disjoint paths, consisting of undirected edges, from $s$ to $t$ where we call two paths are disjoint if they do not share edges. We can assign directions for these
paths, from $s$ to $t$. Now we have $|C|$ disjoint paths from $s$ to $t$. We use $P$ to denoted the set of directed edges in these paths,
\begin{align*}
P=\{(u,v)| (u,v)\   \mathrm{belongs \ to\  some\  path}. \}
\end{align*}

We use the information of $P$ to update the flow $g: E\mapsto \mathbb{Q}$.

We know that for any directed edge $(u,v)\in E$, if $g_{u,v}=1$, we can change the assignement of $g_{u,v}$ and $g_{v,u}$ by letting $g'_{u,v}=\frac{1}{g_{v,u}}$, and $g'_{v,u}=1$ and do not change the rest. $g'$ is still a rational flow and $|g'|=QMC(n_0\mathcal{N})$.

For any sufficiently large $k\in\mathbb{N}$, we construct a flow $f: E\mapsto \mathbb{N}$ of $kn_0m_0\mathcal{N}$.

If the directed edge $(u,v)\in P$, we let
\begin{align*}
f_{u,v}=km_0 g'_{u,v}\ \ f(v,u)=1.
\end{align*}
For $(u,v)\in E\setminus P$, and $ g'_{u,v}=\frac{p}{q}$, $ g'_{v,u}=1$, we let
\begin{align*}
f_{u,v}=p,\ \ f_{v,u}=q.
\end{align*}

One can cerify that $f$ is a flow by observing
\begin{itemize}
\item $f_{u,v}\in\mathbb{N}$ since $m_0 g'_{u,v}\in\mathbb{N}$.
\item Capacity constraint: $km_0g'_{u,v}\leq km_0n_0d_{u,v}$ and $pq\leq kn_0m_0 d_{u,v}$.
\item The conservation condition is preserved during the process: $\mathop{\Pi}\limits_{u:(u,v)\in E} {f_{u,v}} =\mathop{\Pi}\limits_{w:(v,w)\in E} f_{v,u}$ for $v\notin\{s,t\}$.
\end{itemize}
The conservation condition follows from the following observation: If $v$ is not in $P$, then the conservation condition follows from the conservation condition of $g'$; If $v$ is in $P$, then 
\begin{align*}
\mathop{\Pi}\limits_{u:(u,v)\in E} {f_{u,v}} =\mathop{\Pi}\limits_{w:(v,w)\in E} f_{v,u}
\end{align*}
can be obtained from
\begin{align*}
\mathop{\Pi}\limits_{u:(u,v)\in E} {g'_{u,v}} =\mathop{\Pi}\limits_{w:(v,w)\in E} g'_{v,u}.
\end{align*}
by multiplying $(km_0)^r$ from each handsides for approperate integer $r$ due to the face that $P$ is consisting of directed paths.

One can verify that
\begin{align*}
&|f|\\
=&\frac{\mathop{\Pi}\limits_{v:(s,v)\in E} f_{s,v}}{\mathop{\Pi}\limits_{u:(u,s)\in E} f_{u,s}}\\
=&\mathop{\Pi}\limits_{v:(s,v)\in P}\frac{f_{s,v}}{f_{v,s}}\times \mathop{\Pi}\limits_{v:(s,v)\in E\setminus P}\frac{ f_{s,v}}{ f_{u,s}}\\
=&\mathop{\Pi}\limits_{v:(s,v)\in P}{f_{s,v}}\times \mathop{\Pi}\limits_{v:(s,v)\in E\setminus P}\frac{ f_{s,v}}{ f_{v,s}}\\
=&\mathop{\Pi}\limits_{v:(s,v)\in P} km_0 g'_{s,v}\times \mathop{\Pi}\limits_{v:(s,v)\in E\setminus P} g'_{s,v}\\
=&(km_0)^{|C|}\mathop{\Pi}\limits_{v:(s,v)\in P} g'_{s,v}\times \mathop{\Pi}\limits_{v:(s,v)\in E\setminus P} g'_{s,v}\\
=&(km_0)^{|C|}\mathop{\Pi}\limits_{v:(s,v)\in E} g'_{s,v}\\
=&(km_0)^{|C|} |g'|\\
=&(km_0)^{|C|}QMC(n_0\mathcal{N})\\
=&QMC(km_0n_0\mathcal{N}).
\end{align*}
This complete the proof.
\end{proof}

\begin{lem}\label{main-2}
Given quantum tensor network $\mathcal{N} = (G, d, S, T)$, if there exist a flow $f:E\mapsto \mathbb{N}$ such that
\begin{align*}
|f|=QMC(\mathcal{N}),
\end{align*}
then there exists a strict flow $g:E\mapsto \mathbb{N}$ such that
\begin{align*}
|g|=QMC(\mathcal{N}).
\end{align*}
\end{lem}
\begin{proof}
The difference between flow and strict flow is that in the strict flow, there is no incoming flow of the source $s$ and no outgoing flow of the sink $t$.

Let us fix a minimum cut of $\mathcal{N}$, denoted by $C:= (\bar{S}_0\times\bar{T}_0)\cap E$. According to Lemma \ref{flow},
\begin{align*}
QMC(\mathcal{N})=|f|=\frac{\mathop{\Pi}\limits_{\begin{subarray}{c} (v,w)\in E\\v\in \bar{S}, w\in \bar{T}\end{subarray}} f_{v,w}}{\mathop{\Pi}\limits_{\begin{subarray}{c} (r,u)\in E\\ u\in \bar{S}, r\in \bar{T}\end{subarray}} f_{r,u}}\leq \mathop{\Pi}\limits_{\begin{subarray}{c} (v,w)\in E\\v\in \bar{S}, w\in \bar{T}\end{subarray}} f_{v,w}= QMC(\mathcal{N})
\end{align*}
Thus, for directed edge $(v,w)\in C$
\begin{align*}
f_{v,w}=d_{v,w} \ \mathrm{and} \ f_{w,v}=1.
\end{align*}
Consider the integer $\Pi_{(u,v)\in E} f_{u,v}$. According to the fundamental theorem of arithmetic, we have
\begin{align*}
\Pi_{(u,v)\in E} f_{u,v}=\Pi_{i=1}^r p_i^{n_i}
\end{align*}
for prime numbers $p_1<p_2<\cdots<p_r$.

Consider the conservation of flow $f$: $\forall v\in V\setminus\{s,t\}$:
\begin{align}
\mathop{\Pi}\limits_{u:(u,v)\in E} {f_{u,v}} =\mathop{\Pi}\limits_{w:(v,w)\in E} f_{v,w}.
\end{align}
For any $p\in \{p_1,\cdots,p_r\}$, we define
\begin{align*}
p^{\alpha_{u,v}}\mid f_{u,v} \ \ \mathrm{and}\  p^{\alpha_{u,v}+1}\nmid f_{u,v}.
\end{align*}
the conservation of flow $f$ implies that
\begin{align*}
\mathop{\sum}\limits_{u:(u,v)\in E}\alpha_{p:u,v}=\mathop{\sum}\limits_{w:(v,w)\in E}\alpha_{p:v,w}.
\end{align*}

In other words, for any $p\in \{p_1,\cdots,p_r\}$, we obtain a classical flow network $\mathcal{N}_p=(G,\alpha, S,T)$, where $G=(V,E)$ is a directed graph such that for $(u,v)\in E$, we have $(v,u)\notin E$ for all $u,v \in V$. Moreover, this $\mathcal{N}_p$ is already a flow.

If $f_{v,s}=f_{t,u}=1$ for all $(s,v),(u,t)\in E$, $f$ is already strict. Let $g=f$ is enough.

Otherwise, there exists $f_{v,s}>1$ or $f_{t,u}>1$. We will construct $g$ according to $f$.

Case 1: There exists $f_{v,s}>1$. Let $f_{v,s}=\Pi_{j=1}^b q_j^{a_j}$ for distinct prime numbers $q_1<\cdots<q_b$ and $a_j\in\mathbb{N}_{\geq 1}$. It is clear $\{q_1,\cdots,q_b\}\subseteq\{p_1,\cdots,p_r\}$.

For any $q\in \{q_1,\cdots,q_b\}$, let us consider the classical flow network $\mathcal{N}_q=(G,\alpha, S,T)$. The definition of $q$ implies that $(v,s)$ an incoming edge in this $\mathcal{N}_q$ with $1\leq \alpha_{q:v,s}\in \mathbb{N}$. We focus on amount $1$ of this edge $(v,s)$, according to the conservation condition of classical network, there exists an edge $(w,v)$ such that $1\leq \alpha_{q:w,v}\in \mathbb{N}$. Repeat this procedure, we must reach a circle $u_0,\cdots,u_m$ such that for all $0\leq i\leq m-1$,
\begin{align*}
1\leq \alpha_{q:u_m,u_0}\in \mathbb{N}\\
1\leq \alpha_{q:u_i,u_{i+1}}\in \mathbb{N}
\end{align*}
Now, we let
\begin{align*}
\alpha'_{q:u_m,u_0}=\alpha_{q:u_m,u_0}-1\\
\alpha'_{q:u_i,u_{i+1}}=\alpha_{q:u_i,u_{i+1}}-1.
\end{align*}
and leave the rest of $\alpha_{q}$ unchanged.
It is clear that this $\alpha'_{q}$ still forms a valid flow, and
\begin{align*}
\alpha'_{q:u,v}\leq \alpha_{q:u,v}.
\end{align*}
We can use $\alpha'$ instead of $\alpha$ to reconstruct a quantum flow $f'$ for $\mathcal{N}$:
\begin{align*}
f'_{u,v}=f_{u,v}\times q^{\alpha'_{q:u,v}-\alpha_{q:u,v}}.
\end{align*}
The conservation condition is still valid, and the capacity constraint is valid.

According to the discussion of minimum cut, we know $\alpha_{v,w}=0$ for f any directed edge $(v,w)\in C$.
Therefore, $u_0,\cdots,u_m\in\bar{S}$. Therefore,
\begin{align*}
|f'|=|f|.
\end{align*}

Repeat this procedure for each prime factor of each $f_{v,s}>1$, until there is no more circle. We have reach a strict flow $h$ such that
\begin{align*}
|h|=|f|=QMC(\mathcal{N}).
\end{align*}

Case 2: There exists $f_{t,u}>1$. The method is similar to Case 1.

This complete the proof of this lemma.

\end{proof}

According to Lemma \ref{main} and Lemma \ref{main-2}, we have
\begin{thm}\label{qqmaxmin}
For any $\mathcal{N} = (G, d, S, T)$, there are infinitely many $n\in\mathbb{N}$ such that
\begin{align*}
QMF_s(n\mathcal{N})=QMC(n\mathcal{N}).
\end{align*}
\end{thm}

\section{Operational meaning of the quantum maximum flow minimum cut theorem}\label{om}

In this section, we provide an operational interpretation of the quantum maximum flow minimum cut theorem
in last section.
\begin{thm}\label{opt}
Given quantum tensor network $\mathcal{N} = (G, d, S, T)$, if 
\begin{align*}
QMF_s(\mathcal{N})=QMC(\mathcal{N}),
\end{align*}
then there exists a teleportation protocol which transmits $QMC(\mathcal{N})$ dimensional quantum system from the source to the sink.
\end{thm}
\begin{proof}
According to $QMF_s(\mathcal{N})=QMC(\mathcal{N})$, there is a flow $f:E\mapsto \mathbb{N}$ with $f_{v,s}=f(t,u)=1$ for all $(v,s),(t,u)\in E$ and 
\begin{align*}
|f|=\mathop{\Pi}\limits_{v:(s,v)\in E}f_{s,v}=QMC(\mathcal{N}).
\end{align*}
We use the idea in the proof of Lemma \ref{main-2} to construct a teleportation protocol to implement the flow $f$.

For any $p\mid f_{s,v_0}$ for some $(s,v_0)\in E$, we consider $\mathcal{N}_p$, then we can find a sequence 
$s,v_0,\cdots,v_m,t$ such that $p\mid f_{v_i,v_{i+1}}$ and $p\mid f_{v_m,t}$. 

We may find s circle $u_0,\cdots,u_r,u_0$, using the conservation condition, we can always eliminate such circle by modifying
\begin{align*}
f'_{u_i,u_{i+1}}=f_{u_i,u_{i+1}}\times p^{-1},\\
f'_{u_r,u_{0}}=f_{u_r,u_{0}}\times p^{-1},
\end{align*}
and leave the rest of $f$ unchange.
One can verify $f'$ is still a strict flow with $|f'|=QMC(\mathcal{N})$.
Therefore, we can assume there is no such circle in $f$. This means $s\notin\{v_0,\cdots,v_m\}$.

Now, we can teleport the $p$ dimensional quantum system from $s$ to $t$ firstly. The effect is 
\begin{align*}
f''_{s,v_0}=f_{s,v_0}\times p^{-1},\\
f''_{v_i,v_{i+1}}=f_{v_i,v_{i+1}}\times p^{-1},\\
f''_{v_m,T}=f_{v_m,T}\times p^{-1}.
\end{align*}
$f''$ is still a strict flow with $|f''|=QMC(\mathcal{N})\times p^{-1}$.

Repeat this procedure, we can transmit quantum system from $s$ to $t$ until the rest flow $g$ satisfies $g_{s,v}=1$ for all $(s,v)\in E$. The dimension of the transmitted system is
$\mathop{\Pi}\limits_{v:(s,v)\in E}f_{s,v}=QMC(\mathcal{N})$.

According to Lemma \ref{flow},
\begin{align*}
1=\mathop{\Pi}\limits_{v:(s,v)\in E}g_{s,v}=|g|=\mathop{\Pi}\limits_{w:(w,t)\in E}g_{w,t}.
\end{align*}
Then, $g_{w,t}=1$ for all $(w,t)\in E$.

If there is $g_{u_0,u_1}>1$ for some $(u_0,u_1)\in E$. According to the conservation condition of $g$, we can have a circle $u_0,u_1,\cdots,u_m$ such that $q\mid f_{u_i,u_{i+1}}$ and $q\mid f_{u_m,u_0}$ for any $q\mid g_{u_0,u_1}$. This contradicts the assumption that there is no such circle after the circle elimination.

Therefore, the flow $g$ is an empty flow. In other words, we already transmit $QMC(\mathcal{N})$ dimensional quantum system from $s$ to $t$ using teleportation.
\end{proof}

Using this lemma, we can derive the following relation between $QMF$ and $\widetilde{QMF}$.
\begin{lem}\label{relation}
For quantum tensor network $\mathcal{N} = (G, d, S, T)$, if 
\begin{align*}
QMF_s(n\mathcal{N})=QMC(n\mathcal{N}).
\end{align*}
then,
\begin{align*}
\widetilde{QMF}(\mathcal{N})=QMC(\mathcal{N}).
\end{align*}
\end{lem}

\begin{proof}

According to Theorem \ref{opt}, there is a teleportation protocol which transmits $QMC(\mathcal{N})$ dimensional quantum system from the source $s$ to the sink $t$.

We design a new quantum tensor network $\mathcal{M} = (G', d', S', T')$ with $G'=(V',E')$ where
\begin{align*}
V'=V\cup\{s'\}\cup\{t'\},\\
S'=\{s'\},\ \ T'=\{t'\},\\
V'=V\cup \{(s',s)\}\cup\{(t,t')\},\\
d'_{s',s}=d'_{t,t'}=QMC(\mathcal{N}).
\end{align*}
From one hand, we apply the teleportation protocol promised by Theorem \ref{opt}, then the entanglement shared between $s'$, and $t'$ will be a $QMC(\mathcal{N})$ dimensional maximally entangled state.

On the other hand, $\widetilde{QMF}(\mathcal{N})$ is the maximal dimension of the entanglement that can be generated through applying SLOCC on vertices $v\in V\setminus(S\cup T)$. We apply the SLOCC protocol on vertices $v\in V\setminus(S\cup T)$, and obtain a entangled state between $s$ and $t$ with rank $\widetilde{QMF}(\mathcal{N})$. The result is an one-dimensional chain $s'stt'$ where the entanglement rank in $(s',s)$ is $QMC(\mathcal{N})$, in $(s,t)$ is $\widetilde{QMF}(\mathcal{N})$, in $t,t'$ is $QMC(\mathcal{N})$. By assigning tensors on $s$ and $t$, and contract them would result entanglement between $s'$ and $t'$. Then rank is at most $\widetilde{QMF}(\mathcal{N})$. According to the structure of SLOCC, we know that the maximum rank of entanglement between $s'$ and $t'$ over all tensors assignments in $V$ is at most $\widetilde{QMF}(\mathcal{N})$.

Because the teleportation protocol is always an SLOCC protocol, we conclude that
\begin{align*}
QMF_s(n\mathcal{N})\leq \widetilde{QMF}(\mathcal{N}).
\end{align*}

According to Theorem \ref{exist},
\begin{align*}
\widetilde{QMF}(\mathcal{N})\leq QMC(\mathcal{N}),
\end{align*}
we conclude that
\begin{align*}
\widetilde{QMF}(\mathcal{N})=QMC(\mathcal{N}).
\end{align*}
\end{proof}
According to Lemma \ref{main}, Lemma \ref{main-2} and Lemma \ref{relation}, we have
\begin{thm}\label{relation}
For quantum tensor network $\mathcal{N} = (G, d, S, T)$, there exits infinitely many $n$ such that
\begin{align*}
\widetilde{QMF}(n\mathcal{N})=QMC(n\mathcal{N}).
\end{align*}
\end{thm}

\begin{Cor}
For any $\mathcal{N} = (G, d, S, T)$,
\begin{align*}
\lim_{n \to \infty} \frac{\widetilde{QMF}(n\mathcal{N})}{QMC(n\mathcal{N})}=1.
\end{align*}
\end{Cor}
\begin{proof}
According to Theorem \ref{exist}, for all $n$, we have
\begin{align*}
\frac{\widetilde{QMF}(n\mathcal{N})}{QMC(n\mathcal{N})}\leq 1.
\end{align*}

For sufficient large $n=kn_0m_0+r$ with $0<r<n_0m_0$ where $n_0$ and $m_0$ are given in the proof of Lemma \ref{main}.

According to Lemma \ref{main} and Lemma \ref{main-2} and  Lemma \ref{relation}, we have
\begin{align*}
\widetilde{QMF}(kn_0m_0\mathcal{N})=QMF_s(kn_0m_0\mathcal{N})=QMC(kn_0m_0\mathcal{N})=(kn_0m_0)^c QMC(\mathcal{N}).
\end{align*}
where $c$ is an integer determined by $\mathcal{N}$. For connected tensor network, $c\geq 1$.

By the monotonicity property of $\widetilde{QMF}$, we know that
\begin{align*}
\widetilde{QMF}(n\mathcal{N})\geq \widetilde{QMF}(kn_0m_0\mathcal{N}) =(kn_0m_0)^c QMC(\mathcal{N}).
\end{align*}
Therefore,
\begin{align*}
\frac{\widetilde{QMF}(n\mathcal{N})}{QMC(n\mathcal{N})}\geq 
\frac{(kn_0m_0)^c QMC(\mathcal{N})}{n^c QMC(\mathcal{N})}= (\frac{kn_0m_0}{n})^c
\end{align*}
That is,
\begin{align*}
1\geq \frac{\widetilde{QMF}(n\mathcal{N})}{QMC(n\mathcal{N})}\geq  (\frac{kn_0m_0}{n})^c.
\end{align*}
Observe that $k \to\infty$ as long as $n\to \infty$, then we have
  \begin{align*}
\lim_{n \to \infty} \frac{\widetilde{QMF}(n\mathcal{N})}{QMC(n\mathcal{N})}=1.
\end{align*}
\end{proof}

\section{Conclusion \& Acknowledgments}
We present a max-flow min-cut theorem for quantum tensor networks: the maximum amount of flow passing from the source to the sink is equal to the total weight of the edges in a minimum cut after attaching an appropriate dimensional maximally entangled state among each edge of the given quantum tensor network.
 
This work is supported by DE180100156 and DP210102449.

\bibliographystyle{alpha}
\bibliography{qmm}

\end{document}